\documentclass[preprint,11pt]{article}
\usepackage{cite}
\usepackage{amssymb,amsfonts,amsmath,amsthm,amscd,bbold,stackrel,dsfont,mathrsfs}
\usepackage{epic}
\usepackage{graphicx}
\DeclareMathAlphabet{\mathpzc}{OT1}{pzc}{m}{it}

\newtheorem{example}{Example}
\newcommand{\bex}{\begin{example}}
\newcommand{\eex}{\end{example}}

\newcommand{\reals}{\ensuremath{\mathbb{R}}}
\newcommand{\naturals}{\ensuremath{\mathbb{N}}}

\newcommand{\prob}{\ensuremath{\mathbb{P}}}

\def\0{{\tt 0}} 
\def\1{{\tt 1}} 
\def\?{{\tt *}} 
\renewcommand{\mid}{\,|\,}

\newcommand{\E}{{\ensuremath{\tt E}}}

\newcommand{\de}[1] {x_{#1}}

\newtheorem{propo}{Proposition}[section]
\newtheorem{lemma}[propo]{Lemma}

\newtheorem{thm}[propo]{Theorem}

\def\F{{\mathbb F}}

\def\dim{{\rm dim}\, }
\def\Code{{\mathfrak C}}
\def\rank{{\rm rank}}
\def\H{{\mathbb H}}

\def\ve{\varepsilon}
\def\Rho{P}

\def\sTV{\mbox{\tiny \rm TV}}
\def\Ball{{\sf B}}
\def\Tree{{\sf T}}
\def\ux{\underline{x}}

\def\uy{\underline{y}}
\def\uz{\underline{z}}
\def\Lab{{\mathbb L}}
\def\ed{\stackrel{{\mbox{\tiny \rm d}}}{=}}
\def\uX{\underline{X}}
\def\uY{\underline{Y}}
\def\uZ{\underline{Z}}
\def\de{{\rm d}}
\def\Fu{{\sf F}}
\def\eprob{\omega}
\def\lab{{\mathfrak h}}

\begin{document}

\title{Coding for Network Coding}
                                                                        
\author{Andrea Montanari\thanks{Departments of 
Electrical Engineering and Statistics,
Stanford University, Stanford CA-9305, USA}\;\;\;
 and \;\;
R\"udiger Urbanke\thanks{School of Computer and Communication Sciences, EPFL, 1015 Lausanne, CH
\newline
}}

\date{November 19, 2007}
\maketitle
\abstract{
We consider communication over a noisy network under randomized
linear network coding. Possible error mechanism include node- or
link- failures, Byzantine behavior of nodes,  or an over-estimate
of the network min-cut.  Building on the work of K{\"o}tter and
Kschischang, we introduce a probabilistic model for errors.  We
compute the capacity of this channel and we define an error-correction
scheme based on random sparse graphs and a low-complexity decoding
algorithm. By optimizing over the code degree profile, we show that
this construction achieves the channel capacity in
complexity which is jointly quadratic in the number of coded information bits
and sublogarithmic in the error probability.}

\section{Introduction}
Consider a wire-line communication network modeled as a directed
acyclic (multi-)graph with edges of unit capacity. A sources wants
to communicate information to a set of receivers.  If we allow {\em
processing} of information at nodes in the network then the achievable
throughput is in general higher than what can be achieved by
schemes that only allow {\em routing} \cite{ACLY00,LYC03}.  Schemes
that employ processing are referred to as {\em network coding}
schemes.

The standard assumption in the network coding literature is that 
no errors are introduced within the network or, equivalently, that
sufficiently powerful error-correcting codes are employed on the links
at the physical layer. However a number of error sources (e.g.,
malicious or malfunctioning nodes) cannot be neglected. 
We consider a probabilistic model for transmission errors
that builds upon the work of Kschischang and K{\"o}tter \cite{KoK07,SKK07}.
We compute the information theoretic limit on point-to-point communication
for this model (the channel capacity) and
define a coding scheme based on a sparse-graph
construction and a low-complexity iterative decoding algorithm. 
We show that the parameters of the construction can be optimized analytically and, remarkably,
the optimized scheme achieves the channel capacity.
This is the second channel model 
for which iterative schemes can be shown to achieve capacity
(the first one being the binary erasure channel; this was shown in the 
seminal work of Luby, Mitzenmacher, Shokrollahi, Spielman, and
Steman \cite{LMSSS97}).

%
%
\section{Network Coding: Background and Related Work}

Assume that an information source generates 
$h$ symbols per unit time. The integer $h$ is referred to as the source
rate.
Information is encoded at the sender in packets of length $N$ with entries from a finite
field $\F_q$. The network is assumed to be synchronous 
and without delay. As a consequence, packets are aligned at the 
destination at regular time intervals.

The most common scenario studied in this context is a multi-cast one in
which the source aims at communicating the same information
to a set of receivers (distinct nodes in the same network.)
The fundamental theorem of network coding states that
this is possible using network coding (i.e., processing at the nodes) 
if the values of the min-cuts from the source to any of the receivers is at least $h$ \cite{ACLY00}. 
Moreover, {\em linear} network coding suffices \cite{LYC03}. This means that processing
at the nodes can be limited to forwarding packets which are linear (over $\F_q$) 
combinations of incoming packets.
Finally, it is not necessary to choose the local encoding functions at the nodes carefully.
{\em Random} linear combinations are sufficient with probability close to one, 
provided the cardinality of the field is large enough \cite{HKMKE03,KoM04}.
For a general introduction into network coding we refer the reader to \cite{YLCZ05a,YLCZ05b,FrS07}.

The preferred method to implement random linear network codes is
to include ``headers'' in the packets of length $N$ \cite{CWJ03}.
The role of the headers is to ``record'' the coefficients
used in the local encoding functions so that the receiver can 
be oblivious to the network topology and to the specific local encoding functions used.
In more detail, assume that we send $\ell$ packets. The header of each packet
is then an element of $(\F_q)^{\ell}$, where the header of the 
$i$-th source packet, $i \in [\ell]$, is the all-zero tuple, except for an 
identity element at position $i$. 
Recall that nodes forward packets which are linear combinations of the incoming packets.
Therefore, if the header of a packet somewhere in the network reads 
$(\beta_1, \dots, \beta_{\ell})$, $\beta_i \in \F_q$,
then we know that this packet is the linear combination of the $\ell$ original source packets,
where the $i$-th original source packet has ``weight'' $\beta_i$.
The significant advantage of such a scheme is that the receivers can be oblivious
to the topology and the local encoding functions.
Of course we pay some price; if we use headers then only $m=N-\ell$ of the $N$ symbols of
each packet are available for information transmission.
Our subsequent discussion assumes this ``oblivious'' model.

So far we assumed that errors neither occur during transmission nor during
processing. If the channel or the processing are noisy, one
can use coding to combat the noise.
Note that if we stack the $\ell$ source packets of length $N$   
on top of each other then we get an $\ell \times N$ matrix over $\F_q$
whose, lets say, left $\ell \times \ell$ submatrix 
(the collection of headers) is the identity matrix.

Formally, a code $\Code$ is a collection of $\ell\times N$ matrices
with elements in $\F_q$, such that each $M \in \Code$ takes the form 
$M = [\underline{1} \mid \ux]$. Here, $\underline{1}$ is the $\ell\times\ell$ 
identity matrix and $\ux$ is an $\ell\times m$ matrix ($m = N-\ell$). 
We say that $M$ is in {\em normal} form.
The code $\Code$ is thus equivalently described by  a collection of
$\ell\times m$ matrices $\{\ux\}$. The \emph{rate} of the code is
defined as the ratio of the number of information $q$-bits that can be conveyed by the choice of codeword
($\log_q|\Code|$) to the number of transmitted symbols ($N \ell$):
\begin{eqnarray} \label{equ:rate}
R(\Code) = \frac{\log_q|\Code|}{N\ell}\, .
\end{eqnarray}
Before the source packets are transmitted we multiply $M$ from the left by an 
$\ell \times \ell$ random invertible matrix with components in $\F_q$. 
This ``mixes'' the rows of $M$ and ensures that regardless of the network topology and the location
where the errors are introduced, the effect of the errors on the normalized form is uniform.
We then transmit each resulting row as one packet.

Upon transmission of $M$, a ``corrupted'' version $Q$ of the 
codeword is received. Without loss of generality, we assume that
$Q$ is brought back into normal form $Q = [\underline{1} \mid \uy]$
by Gaussian elimination.\footnote{
In principle it might be that the received matrix
cannot be brought in this form because its first $\ell$ columns have rank
smaller than $\ell$. However, within the probabilistic model which we
discuss in the following, the rank
deficiency is small with high probability and can be eliminated
by a small perturbation.}
Following K{\"o}tter and Kschischang \cite{KoK07},
we model the net effect of the transmission-  and the processing-``noise'' 
as a low-rank perturbation
of $\ux$. More precisely, we assume that
\begin{eqnarray}\label{equ:channelmodel}
\uy = \ux + \uz,
\end{eqnarray}
where $\uz$ is an $\ell\times m$ matrix over $\F_q$ of 
$\rank(\uz)=\ell \eprob$, $\eprob \in [0, 1]$. 
We call $\ell \eprob$ the {\em weight} of the error, and $\eprob$
the {\em normalized} weight.

Define the {\em distance}
of two codewords $\ux$ and $\ux'$ as $d(\ux, \ux')=\rank(\ux-\ux')$
and the {\em minimum} distance $d(\Code)$ of the code $\Code$ as the minimum of
the distances between all distinct pairs of codewords. 
The {\em normalized} minimum distance is $\delta(\Code) = d(\Code)/\ell$.
It is shown in \cite{KoK07} that
$d(\cdot, \cdot)$ is a true distance metric; in particular it fulfills
the triangle inequality. Therefore, given a code $\Code$
of minimum distance $d(\Code)$ a simple {\em bounded distance} 
decoder can correct all errors of weight
$s=(d(\Code)-1)/2$ or less. A bounded distance decoder is an algorithm that,
given a received word $\uy$, decodes $\uy$ to the unique word within distance
$s$ if such a word exists and declares an error otherwise. Bounded distance decoders
are popular since a suitable algebraic structure on the code often ensures that
bounded distance decoding can be accomplished with low complexity.

The {\em bounded-distance} error-correcting capability of a code is 
defined as $\eprob(\Code)=d(\Code)/(2\ell)=\delta(\Code)/2$.
K{\"o}tter and Kschischang
showed that the optimal trade-off between 
$R(\Code)$ and $\eprob(\Code)$ is given by an appropriate generalization of
the ``Singleton bound.'' In the limit $N \rightarrow \infty$,
with $\ell=\lambda N$, the maximal achievable rate
for the parameters $\eprob, \lambda \in [0, 1/2]$, call it 
$C_{\text{Singleton}}(\lambda,\eprob)$, is given by
\begin{eqnarray}\label{equ:singletonbound}
C_{\text{Singleton}}(\lambda,\eprob) & = (1-\lambda) (1-2 \eprob).
\end{eqnarray} 
Note that $C_{\text{Singleton}}(\lambda,\eprob)$ is the maximum achievable rate for a guaranteed error correction
in an adversarial channel model. It is also
the maximal achievable rate in a probabilistic setting if we are limited to 
bounded distance decoding.
Remarkably, K{\"o}tter and Kschischang found a  
generalization of Reed-Solomon codes that achieves this bound.

%
%
\section{Main Results}
We are interested in a probabilistic (as opposed to adversarial) channel model.
More precisely, we assume that in (\ref{equ:channelmodel}) the perturbation $\uz$ is 
chosen uniformly at random from all matrices in $(\F_q)^{\ell \times m}$ 
of rank $\ell \eprob$. We assume that the parameters $\lambda$ and $\eprob$ are fixed and
consider the behavior of the channel as we increase $N$. We refer to our channel model
as the {\em symmetric network coding channel} with parameters $\lambda$ and $\eprob$,
denoted by SNC$(\lambda,\eprob)$. 

\begin{propo}[Channel Capacity]
\label{pro:capacityforfixedparameters}
The capacity  of {\rm SNC}$(\lambda,\eprob)$ is
\begin{eqnarray}
C(\lambda,\eprob)  = 1-\lambda-\eprob+\lambda\eprob^2.
\end{eqnarray}
\end{propo}
Discussion: In the definition of capacity we implicitly assume that 
the error probability $\eprob$ is not a function of $N$. 
Depending on the underlying physical error mechanism this may or may not be the case.
Note that for small $\omega$, $C(\lambda,\eprob)  \approx 1-\lambda-\eprob$, whereas
$C_{\text{Singleton}}(\lambda,\eprob)  \approx 1-\lambda-2(1-\lambda)\eprob$.
Fig.~\ref{fig:achievablerates} compares $C(\lambda,\eprob)$ with 
$C_{\text{Singleton}}(\lambda,\eprob)$ and shows the points that 
are achievable according to Theorem~\ref{thm:Main}.

\begin{thm}[Capacity-Achieving Iterative Code Construction]
\label{thm:Main}
For any $\lambda,\eprob\in (0,1)$ such that $(1-\lambda)/\lambda$
is an integer multiple of $\eprob$, any 
$R<C(\lambda,\eprob)$, and any $\pi>0$ there exists an
error correcting code and a decoding algorithm 
that achieves symbol error probability smaller than $\pi$,
with $O(N^4\,\log\log (1/\pi))$ decoding complexity. 
\end{thm}
\begin{figure}[htp]
\begin{center}
\setlength{\unitlength}{1.0bp}%
\begin{picture}(136,120)(-16,-10)
{
\small
\put(0,0){
\put(0,0){\includegraphics[scale=1.0]{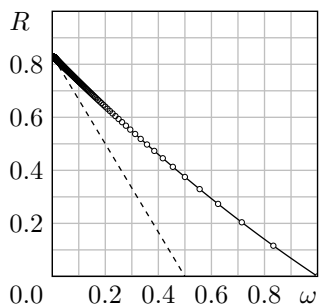}}
\put(100,-10){\makebox(0,0)[rb]{$\omega$}}
{
\multiputlist(0,-10)(20,0)[cb]{$~$,$0.2$,$0.4$,$0.6$,$0.8$}
\multiputlist(-16,0)(0,20)[l]{$~$,$0.2$,$0.4$,$0.6$,$0.8$}
\put(-16,-10){\makebox(0,0)[lb]{$0.0$}}
\put(-16,100){\makebox(0,0)[lt]{$R$}}
}
}
}
\end{picture}
\caption{
Comparison of $C(\lambda,\eprob)$ (solid line) with
$C_{\text{Singleton}}(\lambda,\eprob)$ (dotted line) for $\lambda=1/6$. The points
on the curve $C(\lambda,\eprob)$ that are achievable by the low-complexity iterative scheme
are shown as dots.  \label{fig:achievablerates}}
\end{center}
\vspace{-1cm}

\end{figure}
Discussion: The complexity of the scheme is given as $O(N^4)$. But note that
the number of transmitted information symbols is $N^2 \lambda R$. 
Therefore, if we measure the complexity per transmitted information 
symbol then it is only quadratic.

Note also that the complexity scales much better with 
the target error probability than for usual sparse graph codes 
(where it is at least linear in $\log(1/\pi)$).
%
%
\subsection{Code Construction}
Fig.~\ref{fig:codingscheme} shows our coding scheme.
\begin{figure}[htp]
\begin{center}
\setlength{\unitlength}{0.5bp}%
\begin{picture}(380,120)(0,0)
\put(0,0){\includegraphics[scale=0.5]{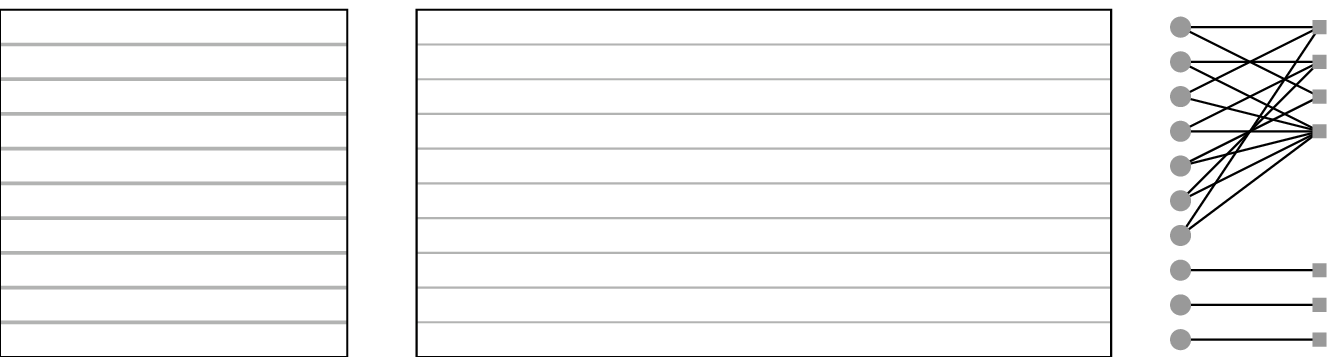}}
\put(5,95){\makebox(0,0){\tiny $1$}}
\put(15,85){\makebox(0,0){\tiny $1$}}
\put(25,75){\makebox(0,0){\tiny $1$}}
\put(35,65){\makebox(0,0){\tiny $1$}}
\put(45,55){\makebox(0,0){\tiny $1$}}
\put(55,45){\makebox(0,0){\tiny $1$}}
\put(65,35){\makebox(0,0){\tiny $1$}}
\put(75,25){\makebox(0,0){\tiny $1$}}
\put(85,15){\makebox(0,0){\tiny $1$}}
\put(95,5){\makebox(0,0){\tiny $1$}}
\put(25,25){\makebox(0,0){0}}
\put(75,75){\makebox(0,0){0}}
{\tiny \multiputlist(125,5)(0,10){$0$,$0$,$0$}}
{\tiny \multiputlist(135,5)(0,10){$0$,$0$,$0$}}
{\tiny \multiputlist(145,5)(0,10){$0$,$0$,$0$}}
{\tiny \multiputlist(155,5)(0,10){$0$,$0$,$0$}}
{\tiny \multiputlist(165,5)(0,10){$0$,$0$,$0$}}
{\tiny \multiputlist(165,5)(0,10){$0$,$0$,$0$}}
{\tiny \multiputlist(175,5)(0,10){$0$,$0$,$0$}}
{\tiny \multiputlist(185,5)(0,10){$0$,$0$,$0$}}
{\tiny \multiputlist(195,5)(0,10){$0$,$0$,$0$}}
{\tiny \multiputlist(195,5)(0,10){$0$,$0$,$0$}}
{\tiny \multiputlist(285,5)(0,10){$0$,$0$,$0$}}
{\tiny \multiputlist(295,5)(0,10){$0$,$0$,$0$}}
{\tiny \multiputlist(305,5)(0,10){$0$,$0$,$0$}}
{\tiny \multiputlist(315,5)(0,10){$0$,$0$,$0$}}
\put(-5,50){\makebox(0,0)[r]{$\ell=\lambda N \left\{\vphantom{\begin{array}{c} a \\ a \\ a \\ a \end{array}}\right.$}}
\put(385,15){\makebox(0,0)[l]{$\left.\vphantom{\begin{array}{c} a \end{array}}\right\} \eprob \ell $}}
\put(385,80){\makebox(0,0)[l]{$\left.\vphantom{\begin{array}{c} \hat{A} \end{array}}\right\} (1-\eprob) \ell \frac{2}{\Rho'(1)}$}}
\put(50,105){\makebox(0,0)[b]{$\overbrace{\hphantom{\hspace{45pts}}}^{\ell=\lambda N}$}}
\put(220,105){\makebox(0,0)[b]{$\overbrace{\hphantom{\hspace{95pts}}}^{m=(1-\lambda) N}$}}
\put(220,50){\makebox(0,0)[c]{$\ux=\left(\begin{array}{c} \!\!x_1\!\! \\ \vdots \\ \!\!x_{\ell}\!\! \end{array}\right)$}}

\end{picture}
\caption{Coding scheme. The last $\eprob \ell$ rows of $\ux$ are zero.
The first $(1-\eprob) \ell$ obey a set of linear constraints represented by
the bipartite graph shown on the right-hand side. \label{fig:codingscheme}}
\end{center}
\vspace{-0.75cm}

\end{figure}
Each row corresponds to a packet of length $N$. 
The $\ell \times \ell$ identity matrix $\underline{1}$ is shown on the left-hand side,
whereas the $\ell \times m$ matrix on the right-hand side represents $\ux$.
Each $\ux$ corresponds to a codeword of $\Code$. 
Not all $\ux$ are allowed.
Here are the constraints that $\ux$ must fulfill to be a codeword.
The bottom $\eprob \ell$ rows are identical to zero.
The top $(1-\eprob) \ell$ rows are constrained by a linear system of equations.
These are indicated by the bipartite graph on the right-hand side,
according to the standard graphical representation used for
low-density parity-check codes \cite{Gal63,RiU07}.
More precisely, we have 
\begin{align}\label{equ:codeconstraints}
\hat{\H}  
\left(
\begin{array}{c}
x_1^T \\ 
\vdots
\\ 
x^T_{(1-\eprob) \ell} \\ 
\end{array}
\right) 
= 0.
\end{align} 
The matrix $\hat{\H}$ has the following structure. Start with a 
``sparse'' $((1-\eprob) \ell r) \times ((1-\eprob) \ell)$  $\{0, 1\}$-valued
matrix $\H$. The matrix $\H$ has exactly $2$ non-zero entries along each 
column.
Further, the fraction of rows that contain exactly $i$ non-zero entries
is equal to $\Rho_i$, where  $\Rho(x)=\sum_{i} \Rho_i x^i$ is
a given {\em degree distribution} (in particular,
it fulfills $\Rho_i \geq 0$ and $\Rho(1)=1$.) In the following, we 
shall say that $\Rho$ has {\em bounded support} if $\Rho_i=0$ for 
$i$ larger than some $n_{\rm max}<\infty$ or, equivalently, if $\Rho(x)$
is a polynomial. 

The matrix $\H$ is represented by the graph. {\em Circles} 
(on the right-hand side in Fig.~\ref{fig:codingscheme}) correspond to
the columns of $\H$ and {\em squares} (on the left-hand side) correspond to the rows of $\H$.
There is an edge between a circle and an edge iff there is a non-zero 
entry at the corresponding row and column of $\H$. 
Following the iterative coding literature, we refer
to the circles as the {\em variable} nodes, to the squares as
the {\em check} nodes,
and we call this graph a {\em Tanner} graph. To get the matrix 
$\hat{\H}$ we ``lift''
$\H$ by replacing each of its non-zero elements by an $m \times m$ invertible matrix with
elements in $\F_q$. We can visualize this by attaching these invertible matrices 
as labels to the corresponding edges.

We claim that for any choice of the matrix $\hat{\H}$ compatible
with the degree distribution $\Rho(x)$ the rate of the code is {\em at least}
\begin{align}\label{equ:rateofourcode}
R(\eprob, \lambda, \Rho) =  (1-\lambda)(1-\eprob) \Bigl(1-\frac{2}{\Rho'(1)} \Bigr).
\end{align}
To see this, note that the matrix $\ux$ is of dimension $\ell \times m$ and has
entries in $\F_q$.
Since the last $\eprob \ell$ rows have to be zero this reduces the degrees of freedom by
$m \eprob \ell$.
Further, there are $m (1-\eprob) \ell \frac{2}{\Rho'(1)}$ linear
constraints, taking away {\em at most} that many further degrees of freedom
(and possibly less because of linear dependencies). We get the claim by
dividing the remaining degrees of freedom by $N \ell$, in accordance with (\ref{equ:rate}).

So far we have explained how to construct a code. We define an {\em ensemble} of codes by
$(i)$ picking a matrix $\H$ uniformly from all matrices that have degree
profile $\Rho(x)$ according to the configuration model and, $(ii)$
picking the labels (the $m \times m$ invertible matrices) for all edges 
uniformly and independently for each edge. 
We denote the resulting ensemble by ${\cal C}(N, \lambda, \eprob, \Rho(x))$.
\vspace{0.1cm}

Discussion: In Fig.~\ref{fig:codingscheme} all linear constraints are on the rows of
$\ux$. An entirely equivalent formulation is to apply the linear constraints to the columns of
$\ux$ instead; i.e., set the last $\eprob m$ columns of $\ux$ to zero and apply a set of linear
constraints on the first $(1-\eprob)m$ columns of $\ux$. All subsequent statements apply 
also to this case and yield identical results if we let $N$ tend
to infinity. For the sake of simplicity, we limit our
discussion to the scheme of Fig.~\ref{fig:codingscheme}. 
In a practical implementation,
however, there can be reasons to prefer one scheme over the other. 
For instance, the iterative decoder discussed in the next section
might be more effective on a larger Tanner graphs. This suggests to use
the construction in Fig.~\ref{fig:codingscheme} if $\ell>m$ and the 
`transposed' one otherwise.
%
%
\subsection{Encoding and Decoding Algorithm}\label{sec:encodinganddecoding}
Assume that the parameters of the model ($N$, $\lambda$, $\eprob$, and $\Rho(x)$) are fixed
and that we have chosen one particular code from the ensemble 
${\cal C}(N, \lambda, \eprob, \Rho(x))$. 
At the source we are given $R N \ell$ symbols over $\F_q$ 
(the information we want to transmit).
We need to map each of these $q^{R N \ell}$ possible information vectors to a distinct codeword $\ux$.
This is the {\em encoding} task.
In principle this can be done by solving a linear system of equations, starting
with (\ref{equ:codeconstraints}).
A brute force approach, however, has complexity $O(N^6)$. Fortunately, one can exploit the sparseness of the matrix $\H$ to reduce the encoding complexity to $O(N^3)$.
The basic idea is to bring $\H$ into upper-triangular form by using only row
and column permutations but no algebraic operations.
As proved in \cite{RiU00encoding}, this can be done with high probability if 
$\Rho''(1)/\Rho'(1)>1$. We will see in 
Section \ref{sec:Proofs}, cf. Lemma \ref{lem:capacityachievingrho},
that this condition is always fulfilled. 
Further details on the efficient implementation of the encoder
will be discussed in a forthcoming publication. We are currently mainly
concerned with the decoding problem.

The receiver sees the perturbed matrix $\uy$. An equivalent description
of our channel model is the following.
Each row of $\uy$ is the result of adding to the corresponding row of 
$\ux$ a uniformly random element of a subspace $W$ of $(\F_q)^m$.
The subspace $W$ is itself uniformly random under the condition
$\dim(W) = \eprob \ell$.\footnote{As discussed in the introduction, 
the underlying physical process is the following: 
we add the headers to the rows of $\ux$; we scramble the rows of $M$
multiplying it by a random invertible matrix in $(\F_q)^{m\times m}$; 
we send the resulting
packets; the channel perturbs these packets; the receiver collects
the perturbed packets, stacks them up to a matrix $Q$, brings the matrix
back into normal form, and ``strips off'' the headers.}

Recall that by assumption the last $\ell \eprob$ rows of $\ux$ are zero.
In fact, in order to achieve reliable transmission we need to 
modify the scheme described so far and set the last $\ell\eprob'$ rows of $\ux$
to $0$, where $\eprob'>\eprob$ is arbitrarily close to $\eprob$.
This modification reduces the rate by a quantity that can be made arbitrarily 
small.
Since the perturbation has dimension $\ell\eprob$,
the last $\ell\eprob'$ rows of $\uy$ will span $W$ with high probability
as $N\to\infty$. A basis of $W$ is then obtained by reducing these rows
via gaussian elimination.

We therefore assume hereafter that $W$ is known and, to avoid
cumbersome notation, we set $\omega'=\omega$.  The decoding task
consists in finding the perturbations for the first $(1-\eprob)
\ell$ rows of $\uy$. If we subtract these perturbations form $\uy$,
we have found $\ux$.  Throughout the description, given two sets of
vectors $U_1$ and $U_2$, we let $U_1+U_2\equiv\{u_1+u_2: \, u_1\in
U_1,\, u_2\in U_2\}$ and, for a given vector $x$, $x+U\equiv
\{x\}+U$.  Finally, given a matrix $\lab\in (\F_q)^{m\times m}$,
$U\lab\equiv\{u\lab:\, u\in U\}$ (vectors are always thought as row
vectors).

We proceed in an iterative fashion. The basic principle is easily understood.
We know that $x_i\in y_i
+ W$. In words, we know that $x_i$ lies in a given affine subspace.
Consider a check node $a$ and, without loss of generality,
let its neighbors be $1, \dots, d$. Let $\lab_{i a}$, $i=1, \dots, d$,
denote the corresponding edge labels. As we discussed earlier, each
such edge label is an $m \times m$ invertible matrix with entries in $\F_q$.
By the definition of the code, $\sum_{i=1}^{d}  x_i\lab_{i,a} = 0$.
In particular, this means that $x_1 \in (\sum_{i=2}^{d} x_i
\lab_{i,a} )\lab_{1,a}^{-1}$.  Since we know that $x_i\in y_i + W$,
this implies that
\begin{align*}
x_1 \in [ (y_2+W)\lab_{2,a} + \cdots +  (y_d+W)\lab_{d,a}]\lab_{1,a}^{-1}\,
.
\end{align*} 
Since we also know that $x_1 \in y_1 + W$, this
implies that 
\begin{align}\label{equ:messagepassin0}
x_1 \in (y_1 + W) \cap \{[(y_2+W)
\lab_{2,a} + \cdots +  (y_d+W)\lab_{d,a}]\lab_{1,a}^{-1}\}\, . 
\end{align}

The actual decoder is most conveniently described (and analyzed)
as a `message passing' algorithm, with messages being sent along the edges
of the Tanner graph.
Messages are affine subspace of $(\F_q)^m$.  They are sent in rounds.
First we send messages from the variable nodes to the check nodes.
We process the incoming messages at the check nodes and then send
messages on all edges from the check nodes to the variable nodes.
This concludes one {\em iteration of message passing}.  

In more
detail, the message sent from variable node $i$ to check node $a$
in the $t$-th iteration is an affine subspace $W_{i\to a}^{(t)}$
of $(\F_q)^m$. If variable node $i$
is connected to check node $a$, let $\bar{a}$ denote the
second check node that is connected to $i$ (recall that each variable
node has exactly two neighbors).
Variable nodes do not perform any non-trivial processing
of the messages, and check-to-variable
node messages coincide with variable-to-check ones
$W_{i\to a}^{(t)}=W_{\bar{a} \to i}^{(t)}$.

For $t=0$ we have $W^{(0)}_{i\to a}=y_i+W$
for all variable nodes $i$ and all check nodes $a$.
Further, let $\partial a$ denote
all neighbors of a check node $a$.  According to the above discussion,
we apply for $t \geq 0$ the recursion
\begin{align}\label{equ:messagepassin} 
W_{i\to a}^{(t+1)} = (y_i + W) 
\cap \Bigl\{\Bigl[\sum_{j \in \partial \bar{a} \setminus i} 
W_{j\to \bar{a}}^{(t)} \lab_{j, \bar{a}} \Bigr]
\lab_{i, \bar{a}}^{-1}\Bigr\}  \, .  
\end{align}
If, after some iterations, 
$\dim(W^{(t)}_{i\to a} \cap W^{(t)}_{i\to \bar{a}})=0$,
then have determined the $i$-th row of $\ux$,
namely $W^{(t)}_{i\to a} \cap W^{(t)}_{i\to \bar{a}} = \{x_i \}$.

Our (main) Theorem~\ref{lem:capacityachievingrho} affirms that.
for given parameters $\lambda$ and $\eprob$, the degree
distribution $\Rho(x)$ can be chosen in such a way that 
the rate of the overall code
approaches the capacity arbitrarily closely and that 
the decoder succeeds with high probability when the packet size 
tends to infinity.

%
\section{Proofs}\label{sec:Proofs}

In the next section we state a few auxiliary lemmas on
the behavior of the message-passing decoder and prove Theorem~\ref{thm:Main}. 
The lemmas are then proved in Section~\ref{sec:ProofLemmas}.
Finally, the capacity of the network coding channel is computed in 
Section \ref{sec:Capacity}.
%
%
\subsection{Auxiliary Results and Proof of the Main Theorem}

To start we can simplify our proof in two manners. First, 
by symmetry of the channel and the message-passing rules, we can
assume that the all-zero matrix $\ux$ was transmitted and we need
only analyze the behavior of the decoder for this case.
Notice that, under this assumption, the messages $W^{(t)}_{i\to a}$
are {\em linear} subspaces (as they must contain the transmitted 
vectors $x_i=0$.)
Second, as we discussed in Section~\ref{sec:encodinganddecoding}, the first
step of the decoding procedure consists of learning the perturbing subspace $W$.
Because of the special structure of the matrix $\ux$ (the last $\omega 
\ell$ rows
are zero) this is accomplished by a simple inspection. We therefore assume in
all that follows that $W$ is known and that the all-zero matrix was transmitted.

Throughout this section we let $\Rho$ be a distribution over the 
integers and let $G$ be a random multi-graph over $\ell(1-r)$ nodes
with degree distribution $\Rho$. The graph $G$ is drawn
according to the configuration model and the code is constructed
from $G$ as described in the previous section.
Since variable nodes have degree $2$, we can think of $G$ either as a multi-graph 
over the check nodes, or as a bipartite graph over check {\em and} variable nodes.

It is also useful to define the `edge perspective' degree distribution
\begin{eqnarray}
\rho_n = \frac{n\Rho_n}{\sum_{n'\ge 0}n'\Rho_{n'}}\, .
\end{eqnarray}

For a uniformly random edge in $G$,  let $W^{(t)}$
be the associated message (that, we recall, is an affine subspace in $(\F_q)^m$).
The key step in the analysis is to notice that the dimension of
$W^{(t)}$ satisfies a simple recursion.

First consider $n-1$ independent and uniformly random linear subspaces
$V_1,\dots,V_{n-1}\subseteq (\F_q)^m$ of dimensions $d_1,\dots,
d_{n-1}$, respectively. Let $V$ be a fixed subspace of $\dim(V)=D$, and define
\begin{eqnarray} 
K_{m,D}^{(n)}(d \mid d_1,\dots,d_{n-1}) \equiv
\prob\{\dim(V\cap(V_1+\dots +V_{n-1})) = d\}, .\label{eq:Kernel}
\end{eqnarray}
The probability kernel $K_{m,D}^{(n)}$ admits an explicit albeit cumbersome expression in terms
of Gauss polynomials. Fortunately,  we do not need its exact description in the following.

We define a sequence of integer-valued random variables
$\{D^{(t)}\}_{t\ge 0}$ recursively as follows. 
For $t=0$ we let $D^{(0)}=\ell\eprob $ identically. For
$t\ge 0$, choose $n$ with distribution $\rho_n$, and draw
$D^{(t)}_1, \dots, D^{(t)}_{n-1}$ iid copies of $D^{(t)}$.
Then, the probability of $D^{(t+1)}=d$, conditioned on the values
$D^{(t)}_1=d_1,\dots,D^{(t)}_{n-1}=d_{n-1}$ coincides with 
Eq.~(\ref{eq:Kernel}) where $D=\ell\eprob$. In formulae,
\begin{eqnarray}
\prob\{D^{(t+1)}=d\}= \sum_{n\ge 1}\rho_n\sum_{d_1\dots d_{n-1}}\!
K_{m,D}^{(n)}(d \mid d_1,\dots,d_{n-1})\,
\prob\{D^{(t)}=d_1\}\cdots \prob\{D^{(t)}=d_{n-1}\}\, .\label{eq:DiscreteDE}
\end{eqnarray}
The sequence $\{D^{(t)}\}$ accurately tracks the dimension of $W^{(t)}$
as stated below.
\begin{lemma}[Density Evolution on a Graph versus Density Evolution on a Tree]\label{lemma:DiscreteDE}
For any degree distribution $\Rho$ with bounded support 
(i.e. such that $\Rho_n=0$ for all $n$ large enough) and any 
$t \in \naturals$
there exists a sequence $\epsilon(\ell,t)$ with $\epsilon(\ell,t)\downarrow 0$
as $\ell\to\infty$, such that, for any $m$, and $\ell$,
\begin{eqnarray}
||\prob\{\dim(W^{(t)})\in\,\cdot\,\}
-\prob\{D^{(t)}\in\,\cdot\,\}||_{\sTV}\le \epsilon(\ell,t)\, ,
\end{eqnarray}
where we recall that $||\prob_X-\prob_{Y}||_{\sTV}=
\sup_{A}|\prob(X\in A)-\prob(Y\in A)|$.
\end{lemma}

Controlling the sequence of random variables $\{D^{(t)}\}_{t\ge 0}$
is quite difficult. Luckily, its behavior simplifies considerably
if we let $m\to\infty$ and consider the scaled
dimensions $D^{(t)}/(\ell\eprob)$. 

More precisely, we define the sequence of 
random variables $\{\xi^{(t)}\}_{t\ge 0}$ with values in $[0,1]$
recursively as follows. We let $\xi^{(0)}=1$ identically. For any 
$t\ge 0$, let $n$ be drawn with distribution $\rho_n$, and 
$\xi^{(t)}_1, \dots, \xi^{(t)}_{n-1}$ be iid copies of $\xi^{(t)}$.
Further, for $a,b,x\in\reals$ with $a\le b$, 
define  $[x]_a^b = \min(\max(x,a),b)$.
Then, the distribution of $\xi^{(t+1)}$ is given by
\begin{eqnarray}
\xi^{(t+1)}\ed \left[\sum_{i=1}^{n-1}\xi^{(t)}_i+1-
\left(\frac{1-\lambda}{\lambda\eprob}\right)\right]_0^{1}\, .\label{eq:XiRec}
\end{eqnarray}
We will prove that the rescaled dimensions $D^{(t)}/(\ell\eprob)$
are accurately tracked by $\xi^{(t)}$.
\begin{lemma}[Density Evolution versus Rescaled Density Evolution]\label{lemma:ContinuousDE}
For any $n_{\rm max}$, $\eprob$, and $\lambda$ there exists 
$\ve>0$ such that, for any degree distribution $\Rho$ with support in 
$[0,n_{\rm max}]$:
\begin{eqnarray}
\lim_{m\to\infty}\prob\{D^{(t+1)}>0\}\le n_{\rm max}\,
\prob\{\xi^{(t)}\ge \ve\}\, .
\label{eq:ContinuousDE}
\end{eqnarray}
\end{lemma}

The previous lemma shows that 
it suffices to consider the behavior of $\xi^{(t)}$ for which we have the
explicit simple recursion (\ref{eq:XiRec}). Even so, finding a degree distribution
$\rho$ which results in codes of large rates and so that  $\xi^{(t)}$ converges to
$0$ for large values of $\delta$, seems challenging. The key to our analysis
is the observation that the recursion (\ref{eq:XiRec}) simplifies significantly
if $(1-\lambda)/\lambda$ is an integer multiple of $\eprob$. In this case
the distribution of $\xi^{(t)}$ trivializes:  $\xi^{(t)}$ only takes on the
values $0$ or $1$ regardless of the degree distribution $\rho$. Density evolution
therefore collapses to a scalar recursion, making it possible to find the optimum degree
distribution $\rho$.

\begin{lemma}[Capacity Achieving Degree Distributions for Rescaled Density Evolution]\label{lemma:KeyLemma}
Let $\lambda,\eprob\in (0,1)$ be such that $(1-\lambda)/\lambda$
is an integer multiple of $\eprob$ and let
$r<C(\lambda,\eprob)/((1-\lambda)(1-\eprob))$.
Then there exists 
$\rho$ with bounded support and $1-2\int_0^1\rho(x)\de x\ge r$,
 and two constants $A>0$,  $\gamma>1$ 
such that, for any $t$, $\ve>0$,
\begin{eqnarray}\label{equ:xitozero} 
\prob\{\xi^{(t)}\ge\ve\}\le \exp\{-A \gamma^t\}\, .
\end{eqnarray} 
\end{lemma}

\begin{proof}[Proof of the Main Theorem \ref{thm:Main}]
Let $\lambda,\eprob, R$ be as in the statement of the theorem and
$r\in (R/((1-\lambda)(1-\eprob))$, $C(\lambda,\eprob)/((1-\lambda)(1-\eprob))$.
We claim that there exists a degree distribution $\Rho$ with 
support in $[0,n_{\rm max}]$,
with $1-2/\Rho'(1)\ge r$ (equivalently, from the edge perspective,
$1-2\int_0^1\rho(x)\,\de x\ge r$) such that the iterative decoder
achieves error probability smaller than $\pi$ in $O(\log\log(1/\pi))$
iterations. Let us check that this indeed proves the 
theorem. As mentioned above, the perturbation subspace 
$W$ (i.e., the linear subspace of $(\F_q)^m$ spanned by 
the rows of $\uz$) can be inferred with high probability
by the last $\eprob'\ell$ of the output $\uy$. This requires Gaussian elimination
of an $m\times(\ell\eprob')$ matrix 
with elements in $\F_q$, which can be accomplished at a cost of $O(N^3)$ operations.

The rest of the codeword $\ux$ is decoded by message passing. 
Each iteration requires updating $O(N)$ messages (because $\Rho$
has bounded support). Each update, cf. Eq.~(\ref{equ:messagepassin}),
requires finding a basis for a space spanned by, at most 
$(\ell\eprob)n_{\rm max}$ vectors in $(\F_q)^m$. This can be done, again 
via Gaussian elimination, in $O(N^3)$ operations. We thus get
$O(N^4)$ operations per iteration. Since running 
$O(\log\log(1/\pi))$ iterations achieves error probability smaller than
$\pi$, this implies the thesis.

Let us now prove this claim. 
First, we fix the degree distribution in such a way that Lemma
\ref{lemma:KeyLemma} holds for some $A>0$, $\gamma>1$.
We let $t_*(\pi)= O(\log\log(1/\pi))$  be such that 
$\prob\{\xi^{(t)}\ge \ve\}\le \exp\{-A\gamma^t\}\le \pi/(3n_{\rm max})$
for any $t\ge t_*(\pi)$. Then, for any fixed
$t\ge t_*(\pi)$ the decoded error
probability is upper bounded by $\pi$ if $N$ is large enough.

Indeed, $\ve$ can be chosen in such a way that Lemma \ref{lemma:ContinuousDE}
holds and therefore, for $m$ large enough,
$\prob\{D^{(t+1)}>0\}\le n_{\rm max}\, \pi/(3n_{\rm max}) +\pi/3\le 2\pi/3$.
The $i$-th row of codeword $\ux$ is decoded correctly if
any of the two messages $W^{(t+1)}_{i\to a}$ or $W^{(t+1)}_{i\to\bar{a}}$
has dimension $0$. Therefore, the symbol error probability
is upper bounded by $\prob\{\dim(W^{(t+1)})>0\}$.
By Lemma \ref{lemma:DiscreteDE}, for $\ell$ large enough, this is at most
$\prob\{D^{(t+1)}>0\}+\pi/3\le \pi$, which proves the theorem.
\end{proof}
%

\subsection{Proofs of Lemmas}\label{sec:ProofLemmas}

\begin{proof}[Proof of Lemma \ref{lemma:DiscreteDE}]
The proof is based on the `density evolution' technique \cite{RiU07},
and on some remarks that allow to simplify the resulting distributional
recursion. A similar result appeared already in the context of erasure decoding
for non-binary codes \cite{RaU05}: in order to be self-contained we nevertheless
sketch the proof here.

Let $\vec{e}$ be a uniformly random directed edge in $G$ and let
$W^{(t)}$ the associated message after $t$ iterations
of the message-passing algorithm. Denote by $\Ball(\vec{e},t)$
the `directed neighborhood' of $\vec{e}$ with radius $t$, i.e.,
the induced sub-graph containing all non-reversing walks in $G$ 
of length at most $t$ that terminate in $\vec{e}$. We regard this
as a labeled graph with variable node labels given by the received 
vectors and edge labels by the $m\to m$ matrices that define the code. 
It is well known that such a neighborhood
converges to a (labeled) Galton-Watson tree $\Tree(t)$.

More precisely, $\Tree(t)$ is a $t$-generations tree rooted in a directed edge
$\vec{e}_{\sf T}$ and with offspring distribution $\rho_n$.
We have
\begin{eqnarray}
|| \prob\{ \Ball(\vec{e},t)\in\,\cdot\,\}- \prob\{ \Tree(t)\in\,\cdot\,\}
||_{\sTV}\le \epsilon(\ell,t)\, ,
\end{eqnarray}
for some $\epsilon(\ell,t)$ as in the statement of Lemma~\ref{lemma:DiscreteDE}.

Note that the message $W^{(t)}$ is a function only of the neighborhood
$\Ball(\vec{e},t)$.  Suppose that we apply the message-passing algorithm
to $\Tree(t)$ and let $W^{(t)}_{\Tree}$ be the message passed
through the root edge after $t$ iterations. 
It follows from the definition of total variation distance that
\begin{eqnarray}
||\prob\{\dim(W^{(t)})\in\,\cdot\,\}-\prob\{\dim(W^{(t)}_{\Tree})
\in\,\cdot\,\}||_{\sTV}
\le \epsilon(\ell,t)\, .
\end{eqnarray}

The proof is completed by showing that $\dim(W^{(t)}_{\Tree})$
is distributed as the random variable $D^{(t)}$ defined recursively by
Eq.~(\ref{eq:DiscreteDE}). First, note that
$W^{(t)}_{\Tree}$ is a uniformly random subspace, conditional on its
dimension $\dim(W^{(t)}_{\Tree})$. This follows from the message-passing update
rule (\ref{equ:messagepassin}) together with the remark that, given any fixed 
subspace $W_*$ and a uniformly random full-rank $m\times m$ matrix
$\Lab$, $\Lab W_*$ is a uniformly random subspace with the same 
dimension as $W_*$. 

We prove that $\dim(W^{(t)}_{\Tree})$ is distributed 
as $D^{(t)}$ by recursion. The statement is true for 
$t=0$ by definition of our channel model. 
Consider the tree $\Tree(t+1)$ and condition on the offspring 
number at the root $n-1$.  Denote by $W^{(t)}_{\Tree,1},\dots,
W^{(t)}_{\Tree,n-1}$ the corresponding messages towards the root 
and condition on $\dim(W^{(t)}_{\Tree,1}) = d_1$,\dots,
$\dim(W^{(t)}_{\Tree,n-1})= d_{n-1}$. Then the distribution 
of $\dim(W^{(t+1)}_{\Tree})$ is given by the kernel
(\ref{eq:Kernel}) with $D=\ell\eprob$ by uniformity of the subspace.
The claim follows from the  fact that $W^{(t)}_{\Tree,1}$,\dots,
$W^{(t)}_{\Tree,n-1}$ are iid because of the tree structure.
\end{proof}

In the proof of Lemma~\ref{lemma:ContinuousDE} we require an estimate of
the probability that true density evolution deviates significantly
from the the rescaled density evolution.
\begin{propo}[Deviations from Asymptotic Density Evolution]\label{propoDimAlgebra}
Let $V_1$ be a subspace of dimension $d_1$ in $\F_q^m$,
and  $V_2$ a uniformly random subspace of dimension $d_2$.
Define $d_{1}\odot d_{2}\equiv \max(0,d_1+d_2-m)$,
and  $d_{1}\boxplus d_2\equiv \min(m,d_1+d_2)$. Then
\begin{eqnarray}
\prob\{d_{1}\odot d_2\le \dim(V_1\cap V_2) < d_{1}\odot d_2+k\}
\ge 1- q^{-k-\max(0,m-d_1-d_2)}\, ,\label{eq:InterDimensions}\\
\prob\{d_{1}\boxplus d_2-k\le \dim(V_1+ V_2) < d_{1}\boxplus d_2\}
\ge 1- q^{-k-\max(0,m-d_1-d_2)}\, .\label{eq:SumDimensions}
\end{eqnarray}
Further, let $V$ be a subspace of dimension $d$ and let
$V_1,\dots,V_{n-1}$ be uniformly random subspaces of dimensions 
(respectively) $d_!,\dots,d_{n-1}$ and
$d \equiv [d_1+\cdots+d_{n-1}+d-m]^{d}_0$.
Then 
\begin{eqnarray}
\prob\{|\dim((V_1+\cdots+V_{n-1})\cap V)-d|\ge k\}\le n\, q^{-k/n}\, .
\label{eq:BoundDim}
\end{eqnarray}
\end{propo}
\begin{proof}
Notice that Eq.~(\ref{eq:SumDimensions}) follows from 
Eq.~(\ref{eq:InterDimensions}) together with the identity
$\dim(V_1+V_2) = d_1+d_2-\dim(V_1\cap V_2)$. Further 
$\dim(V_1\cap V_2)\ge d_1 \odot d_2$ for any two subspaces
$V_1$, $V_2$ of the given dimensions. 

We are left with the task of
bounding the probability of $\dim(V_1\cap V_2)\ge d_1 \odot d_2+k$.
Notice that this event is identical
to $|V_1\cap V_2|\ge  q^{d_1 \odot d_2+k}$
(we denote by $|S|$ the cardinality of the set $S$). By the Markov inequality 
we have
\begin{eqnarray}
\prob\{\dim(V_1\cap V_2)\ge d_1 \odot d_2+k\}\le q^{-k-d_1 \dot d_2}
\E|V_1\cap V_2| =  q^{-k-d_1 \odot d_2}\, q^{d_1+d_2-m},
\end{eqnarray}
where the equality on the right-hand side follows by multiplying the number
of vectors in $V_1$ (that is $q^{d_1}$) with the probability that 
one of them belongs to $V_2$ (by uniformity this is $q^{-m+d_2}$).

Eq.~(\ref{eq:BoundDim}) follows by applying the previous bound recursively.
bounds. Explicitly, we define $W_1= V_1$, $W_2=W_1+ V_2$,
\dots, $W_{n-1}= W_{n-2}+V_{n-1}$, and $W_{n}=W_{n-1}\cap V$. 
The corresponding (typical) dimensions are $c_1=d_1$,
$c_2= c_1\boxplus d_2, \dots, c_{n-1}=c_{n-2}\boxplus d_{n-1}$,
$c_n = c_{n-1}\odot d_n = d$. By the union bound, with probability 
at least $1-n \, q^{-k/n}$ we have 
$|\dim(W_n)-(d_n\odot \dim(W_{n-1}))|\le k/n$ and 
$|\dim(W_{i})-(d_{i}\boxplus \dim(W_{i-1})|\le k/n$
for $i\in\{2,\dots,n-1\}$. The thesis follows by the triangle inequality.
\end{proof}

\begin{proof}[Proof of Lemma \ref{lemma:ContinuousDE}]
We will first prove that there exists a coupling between $D^{(t)}$
and $\xi^{(t)}$ such that
$|D^{(t)}-(\ell\eprob) \xi^{(t)}|\le \ell \ve $ with high probability 
as $\ell$, $m\to\infty$ (with $\lambda$, $\eprob$ fixed). 
Subsequently, we shall prove that this claim implies the thesis.

The coupling is constructed recursively. For $t=0$ we have $D^{(0)}=
(\ell \eprob) \xi^{(t)}=\ell\eprob$ deterministically. This defines
the coupling of $D^{(t)}$ and $\xi^{(t)}$ for $t=0$. 
Assume we have shown how to construct a coupling of $D^{(t)}$
and $\xi^{(t)}$ for some $t \in \naturals$. To define the coupling for $t+1$
we  draw an integer $n$ with distribution $\rho_n$. We then generate $n-1$ coupled pairs
($D^{(t-1)}_i, \xi^{(t-1)}_i)$. From those we generate a coupled pair
($D^{(t)}_i, \xi^{(t)}_i)$ via the recursions (\ref{eq:DiscreteDE}) and
(\ref{eq:XiRec}), respectively. 

In order to prove the claim it is sufficient to show the following.
If $V_1,\dots, V_{n-1}$ are uniformly random subspaces of dimensions
$(\ell\eprob)\xi_1,\dots, (\ell\eprob)\xi_{n-1}$ in
$\F_q^m$, and if $V$ has dimension $(\ell\eprob)$, then, with high probability,
$|\dim((V_1+\cdots+V_{n-1})\cap V)-(\ell\eprob)\xi|\le \ell\ve$
for any $\ve>0$. This in turns follows from Proposition 
\ref{propoDimAlgebra} (Eq.~(\ref{eq:BoundDim})) together with the 
observation that the degree $n$ is bounded.

Let us now consider the thesis of the lemma, Eq.~(\ref{eq:ContinuousDE}).
We can assume without loss of generality that $n_{\rm max}\ge 1$
and $m>\ell\eprob$, whence $1-\lambda>\lambda\eprob$ follows.
Let $n_{\rm max}\ge 2$ be the largest integer in the support of $\rho_n$
and take $\ve>0$ small enough so that $2(n_{\rm max}-1)\ve
\le (1-\lambda)/(\lambda\eprob)-1-\gamma$ for some $\gamma>0$.
Draw $n_{\rm max}$ iid copies of $D^{(t)}$, denoted
$D^{(t)}_1,\dots, D^{(t)}_{n_{\rm max}}$. Since under the coupling  
$|D^{(t)}-(\ell\eprob) \xi^{(t)}|\le \ell \ve $ with high probability,
\begin{eqnarray}
\prob\left\{\max\{D^{(t)}_1,\dots, D^{(t)}_{n_{\rm max}}\}\ge 2\, \ve \,
(\ell\eprob)\right\}\leq n_{\rm max}\prob\{\xi^{(t)}\ge \ve\}+o_m(1)\, .
\label{eq:UnionBoundSpaces}
\end{eqnarray}
Now draw $n$ with distribution $\rho_n$ and $D^{(t+1)}$ conditional
on $D^{(t)}_1,\dots,D^{(t)}_{n-1}$ according to the kernel 
(\ref{eq:Kernel}). Namely, $D^{(t+1)}$ is the dimension of
$V\cap(V_1+\dots +V_{n-1})$ when $\dim(V)=\ell\eprob$ and 
$V_1,\dots,V_{n-1}$ are uniformly random subspaces of $\F_q^m$
with dimensions $D^{(t)}_1,\dots,D^{(t)}_{n-1}$.

Let $W\equiv V_1+\dots +V_{n-1}$. Then $W$ is uniformly random conditioned
on its dimension
$\dim(W) \le D^{(t)}_1+\cdots+D^{(t)}_{n-1}\le 
\ell(1-\lambda-\lambda\eprob)/\lambda-\ell\eprob\gamma$
with probability lower bounded as in Eq.~(\ref{eq:UnionBoundSpaces}).
Assume this to be the case.
By Proposition \ref{propoDimAlgebra}, Eq.~(\ref{eq:InterDimensions}),
and recalling that $m=\ell(1-\lambda)/\lambda$, 
the probability that $D^{(t+1)}=\dim(V\cap W)>0$ is at most
$q^{-\ell\gamma\eprob}$. This proves the thesis.
\end{proof}

In order to prove our last auxiliary result, Lemma \ref{lemma:KeyLemma},
we need some algebraic properties of the edge-perspective capacity-achieving
degree distribution and of the corresponding generating function:
\begin{eqnarray}
\rho^*_k(x) = \sum_{i = k+1}^{\infty} \frac{k-1}{(i-1)(i-2)}\, x^{i-1}
\equiv\sum_{i=0}^{\infty}\rho^*_{k,i}\, x^{i-1}\, .
\end{eqnarray}

\begin{lemma}[Basic Properties of Capacity-Achieving Degree Distribution]
\label{lem:capacityachievingrho}
Let $k \in \naturals$ and define  
$f_{k, i}(\alpha) = \sum_{j =k}^{i-1} \binom{i-1}{j} \alpha^j 
(1-\alpha)^{i-1-j}$.
Then $\rho^*_k(1)=1$, $\de \rho^*_k(x)/\de x\mid_{x=1} \geq k$, 
$\int_0^1 \! \rho^*_k(x) \de x = 1/(2 k)$,  
and $\sum_{i} \rho^*_{k, i} f_{k, i}(\alpha) =\alpha$.
\end{lemma}
\begin{proof} 
By a reordering of the terms in the sum,
\begin{align*} 
\rho^*_k(1)=\lim_{j\to\infty}
\sum_{i = k+1}^{k+j} \frac{k-1}{(i-1)(i-2)} = \lim_{j\to\infty}
\sum_{i=k+1}^{k+j}
\left(\frac{k-1}{i-2}-\frac{k-1}{i-1}\right) = \lim_{j\to\infty}
\left(1-\frac{k-1}{k+j-1}\right)=1\,.
\end{align*}

In a similar  manner, we have
\begin{align*}
\int_0^1 \! \rho^*_k(x) \de x =
\lim_{j\to\infty}\sum_{i \geq k+1}^{k+j} \frac{k-1}{ i (i-1)(i-2)} = 
\lim_{j\to\infty}\sum_{i \geq k+1}^{k+j}\left( \frac{k-1}{2(i-1)(i-2)}- 
\frac{k-1}{2i(i-1)}\right) = \frac{1}{2k}\, .
\end{align*}
The claim $\int_0^1 \! \rho^*_k(x) \de x = 1/(2 k)$ follows since $\rho^*_k(x)=1$, 
$\rho_{k, i}^* \geq 0$,  and
since $\rho(x)$ only contains powers of $x$ of at least $k$.
In order to prove the last assertion we recall the identity \cite{Wil94}
\begin{eqnarray}
\sum_{n=i}^{\infty}\binom{n}{i} x^n = \frac{x^i}{(1-x)^{i+1}}\, .
\label{eq:Summation}
\end{eqnarray}
We then obtain
(here $\bar{\alpha}\equiv (1-\alpha)$): 
\begin{align*}
\sum_{i} \rho_{k, i}^* f_{k, i}(\alpha) 
= & \sum_{i\geq k+1}   \frac{k-1}{(i-1)(i-2)} \sum_{j=k}^{i-1} \binom{i-1}{j} \alpha^j \bar{\alpha}^{i-1-j} 
=  (k-1) \sum_{j= k}^{\infty} \left(\frac{\alpha}{\bar{\alpha}} \right)^j \sum_{i=j+1}^{\infty}  
\frac{\binom{i-1}{j} \bar{\alpha}^{i-1}}{(i-1)(i-2)} \\
= & (k-1) \sum_{j\geq k}^{\infty} \left(\frac{\alpha}{\bar{\alpha}} \right)^j 
\frac{\alpha^{1-j} \bar{\alpha}^j}{j (j-1)}
=  (k-1) \sum_{j= k}^{\infty}  \frac{\alpha}{j (j-1)} = \alpha\, ,
\end{align*}
where we applied the identity obtained by integrating 
Eq.~(\ref{eq:Summation}) twice with respect to $x$. 
\end{proof}

\begin{proof}[Proof of Lemma \ref{lemma:KeyLemma}]
Let $k= (1-\lambda)/(\lambda \eprob)$, $k \in \naturals$.
Then $C(\lambda,\eprob)/((1-\eprob)(1-\lambda))= 1-1/k$.

It is clear from the recursive definition (\ref{eq:XiRec})
together with the initial condition $\xi^{(0)}=1$ that, for any
$t\ge 0$, $\xi^{(t)}$ only takes values
$0$ and $1$. Let $\alpha_t\equiv \prob\{\xi^{(t)}=1\}$.
Then $\alpha_0=1$, and  Eq.~(\ref{eq:XiRec}) implies that
\begin{eqnarray}
\alpha_{t+1} = \sum_{n=k+1}^{\infty}\rho_{n}\, f_{k,n}(\alpha_t) \equiv 
\Fu_{k,\rho}(\alpha_t), 
\end{eqnarray}
where $f_{k,n}(\alpha)$ is defined as in the statement of Lemma
\ref{lem:capacityachievingrho} (note that $f_{k,k}(\alpha)\equiv 0$). 
We claim that for any 
$r<1-1/k$ there exists an edge-perspective degree distribution $\rho$ of bounded support such that:
(i) $1-2\int_0^1\!\rho(x)\,\de x\ge r$; (ii) $\Fu_{k,\rho}(\alpha)<\alpha$
for any $\alpha\in (0,1]$; (iii) $\Fu_{k,\rho}(\alpha)=O(\alpha^k)$
as $\alpha\downarrow 0$. Then the lemma
follows by standard calculus, with $\gamma\in (1,k)$ and $A$ sufficiently small.

In order to exhibit such a degree distribution, fix 
$b \in \naturals$, $b \geq k$, and define $\rho(x) = 
\sum_{i = k} \rho_{i} x^{i-1}$, where $\rho_i=0$ except for
$\rho_{i}=\rho^*_{k, i}$, $i = k+1, \dots, b$, and $\rho_{k}=1-\sum_{i=k+1}^b \rho^*_{k, i}$. Then
\begin{eqnarray}
\int_0^{1}\rho(x)\, \de x=\sum_{i=k}^{b}\rho_i/i= 
\sum_{i=k}^{b}\rho_{k,i}^*/i+\sum_{i=b+1}^{\infty}\rho_{k,i}^*/k\, .
\end{eqnarray}
By Lemma \ref{lem:capacityachievingrho} the right-hand side
converges to $1/(2k)$ as $b\to\infty$. Therefore we can chose $b$ large 
enough so that claim (i) above is fulfilled.

Consider now claim (ii). We write
\begin{align*}
\Fu_{k,\rho}(\alpha)=
\sum_{i=k+1}^{b} \rho_{i} f_{k, i}(\alpha) = \sum_{i=k+1}^{\infty} 
\rho_{k, i}^* f_{k, i}(\alpha)- \sum_{i=b+1}^{\infty} 
\rho_{k, i}^* f_{k, i}(\alpha)= \alpha - \sum_{i=b+1}^{\infty} 
\rho_{k, i}^* f_{k, i}(\alpha),
\end{align*}
where the last identity follows from Lemma~\ref{lem:capacityachievingrho}.
The claim is implied by the remark that $f_{k,i}(\alpha)>0$ 
for $i\ge k+1$ and $\alpha\in (0,1]$.

Finally, claim (iii) is a consequence of the fact that
$f_{k,i}(\alpha) = \binom{i-1}{k}\alpha^k+O(\alpha^{k+1})$ together
with $i\le b$.
\end{proof}

%
%
\subsection{Capacity}\label{sec:Capacity}

\begin{proof}[Proof of Proposition \ref{pro:capacityforfixedparameters}]
By standard information-theoretic arguments \cite{CoT91}, 
the channel information capacity is given by
\begin{eqnarray}
C(\eprob,\lambda) = \lim_{N\to\infty, \, \ell = N\lambda}\frac{1}{N\ell}\, 
\sup_{\prob_{\uX}}\, I(\uX;\uY)\, .
\end{eqnarray}
Here $I(\uX;\uY) = \sum_{\ux,\uy}\prob_{\uX,\uY}(\ux,\uy)
\log\{\prob_{\uX,\uY}(\ux,\uy)/\prob_{\uX}(\ux)\prob_{\uY}(\uy)\}$
is the \emph{mutual information} between $\uX$ and $\uY$ and 
the supremum is taken over all possible input distributions.

Writing the mutual information in terms of entropy and conditional entropy,
and using our channel model (\ref{equ:channelmodel}), we have 
$I(\uX;\uY) = H(\uY)-H(\uY|\uX) = H(\uY)-H(\uZ)$. Since $H(\uZ)$
does not depend on the input distribution, the mutual information is maximized 
when the latter is uniform. This implies that the output 
is uniform as well, and we get $H(\uY) = \log(q^{m\ell})$.

Finally, $H(\uZ)$ is the logarithm of the number $A(s,\ell,m)$ 
of $\ell\times m$ matrices of rank $\rank(\uZ)=\ell\eprob\equiv s$.
We have $A(s,\ell,m) = q^{m\ell}\prob_0\{\rank(\uZ)=s\}$
where $\prob_0$ denotes probability with respect to a uniformly 
random matrix $\uZ$. Assume without loss of
generality that $\ell,m\ge s$.  
If $z_1,\dots, z_{\ell}$ be the lines of $\uZ$,
then the first $s$ lines are independent with probability 
$(1-q^{-\ell})(1-q^{-\ell+1})\cdots(1-q^{-\ell+s})\ge 1-sq^{-\ell+s}$. 
then the space .  
\begin{eqnarray}
A(s,\ell,m) \ge 
q^{m\ell}\, \prob\{z_{s+1}\dots z_{\ell}\in\span(z_1\dots z_s),
\, \rank(z_1\dots z_s)=s\}\ge q^{m\ell}\, q^{-(\ell-s)(m-s)}(1-sq^{-\ell+s})\, .
\end{eqnarray}
On the other hand  $\prob_0\{\rank(\uZ)=s\}$ is upper bounded by 
summing over all subsets of $s$ lines (there are $\binom{\ell}{s}\le 2^{\ell}$
such subsets), the probability that such lines are independent 
and that the other lines are in the span generated by these. Such an upper
bound is at most $2^{\ell}$ larger than the above lower bound.
By taking $N\to\infty$ with $\ell = N\lambda=N-m$, $\lambda\in (0,1)$
and $\eprob\in(0,\min(1,(1-\lambda)/\lambda$ we get
\begin{eqnarray}
H(\uZ) =\log A(s,\ell,m) = N\ell\, (\eprob+\eprob^2\lambda)+O(N)\, .
\end{eqnarray}
Therefore $I(\uX;\uY)=H(\uY)-H(\uZ) = N\ell(1-\lambda-\eprob+\eprob^2\lambda)
+O(N)$ whence the thesis follows.
\end{proof}

\bibliographystyle{siam} 
\newcommand{\SortNoop}[1]{}

\appendix
\end{document}